\documentclass[12pt]{article}
\usepackage{amsfonts}
\usepackage{eurosym}
\usepackage{graphicx}
\usepackage{booktabs}
\usepackage{float}
\usepackage{subfig}
\usepackage{amssymb}
\usepackage{comment}
\usepackage{amsmath}
\usepackage{booktabs}
\usepackage{multirow}
\usepackage{enumerate}
\usepackage{color}
\usepackage[authoryear]{natbib}
\usepackage{hyperref}
\usepackage{bbm}
\usepackage{setspace}
\usepackage[title]{appendix}

\definecolor{airforceblue}{rgb}{0.36, 0.54, 0.66}
\definecolor{battleshipgrey}{rgb}{0.52, 0.52, 0.51}
\definecolor{charcoal}{rgb}{0.21, 0.27, 0.31}
\hypersetup{colorlinks=true, linkcolor=airforceblue, citecolor=airforceblue, urlcolor=battleshipgrey}

\newtheorem{assumption}{{\bf \sc Assumption}}

\newtheorem{lemma}{{\bf \sc Lemma}}

\newtheorem{proposition}{{\bf \sc Proposition}}

\def\eproof{\hbox{\hskip3pt\vrule width4pt height8pt depth1.5pt}}
\begin{document}

	\title{Social Learning with Endogenous Information and the Countervailing Effects of Homophily
	}
	\author{Yunus C. Aybas\thanks{Department of Economics, Texas A\&M University.} \and Matthew O. Jackson\thanks{%
			Department of Economics, Stanford University. 
			Jackson is also an external faculty member at the Santa Fe Institute.
			We gratefully acknowledge support under NSF grants SES-1629446 and SES-2018554.
	}}

	\date{June 2026}
	\maketitle

	\begin{abstract}
		People learn about beneficial opportunities and actions by observing the experiences of their friends. We model how homophily---the tendency to associate with similar others---affects both the endogenous quality and diversity of the information accessible to decision makers.  Homophily has countervailing implications.  Observing the payoff of another person is more informative the more similar that person's payoff function is to that of the decision maker. However, if agents with similar preferences herd on the same action, then homophily can lead agents to miss out on knowledge of the payoffs of other actions. We show how network connectivity influences this tradeoff and affects the endogenous collection of information. Although homophily hampers learning in sparse networks, it enhances learning in sufficiently dense networks.

		
	\end{abstract}

	\thispagestyle{empty}
	
	\setcounter{page}{0} \newpage
	
	\section{Introduction}

	\cite{chettyetal2022I} find that ``economic connectedness''  emerges as a predictor of economic mobility in the United States and stands out compared to a broad set of types of social capital. Specifically, a lack of friendships across economic divides predicts that people growing up in poor households are significantly more likely to stay poor.  Moreover, \cite{chettyetal2022I} find that once such homophily is accounted for, the prominent relationship between inequality and immobility is mediated.\footnote{The relationship between inequality and immobility is known as the Great Gatsby Curve: see \cite{corak2016,jackson2019,jackson2021} for background.}  Places with lower levels of connection between rich and poor are both more unequal and have lower economic mobility.
	Moreover, \cite{power2026endow} find that communities, drawn from a data set of sites around the globe, with lower levels of economic connectedness have higher levels of inequality.
	One hypothesis from such findings is that information bridging capital---friendships across groups that might have different information---leads to better social learning about advantageous opportunities, which can lead to both greater equality and mobility.  In this paper, we build a model that allows us to study that and related hypotheses.

	Given that people communicate most frequently with people with similar preferences and behaviors---due to homophily---they can end up only learning about a limited set of options, thus missing out
	on valuable opportunities.   Different groups might not only end up with differing levels of education and outcomes within a generation, but this can persist across generations.
	Therefore, it is important to understand how, when, and why homophily induces systematic differences in people's beliefs and behaviors across different groups, and how those differences can persist dynamically.
	
	On the plus side, homophily can improve the relevance of information that individuals obtain from friends' experiences. Observing that someone with similar preferences, background and skills has a positive experience---such as getting into and enjoying a particular graduate education program---can provide more relevant information than observing the outcome for someone with different preferences, background and skills.
	For example, \cite{porter2020gender} show that female role models are influential in determining other female students' major choices.
	\cite{sorensen2006social} shows that employees learn more about health plan choices from colleagues with similar backgrounds, \cite{conley2010learning} show that farmers learn more about agricultural techniques when their peers have a similar wealth level, and \cite{malmendier2022information} examine the influence of who communicates information on decisions of vaccine adoption among other things (see also the
	seminal work by \cite{katzl1955}). Thus, homophily can improve the \textit{quality} of information that someone gets from their network.
	
	On the minus side, homophily can lead a person to learn only about a limited set of choices.  
	Moreover, this can lead to inefficiencies via the impact it has on endogenous information generation.
	If a person's friends are predominantly from a group with similar preferences and information, and they herd on the same choice, then that person ends up only learning about that choice and not about other choices.
	For example, many investors invest disproportionately in domestic equities, ignoring the benefits of diversifying into foreign equities \citep*{frenchp1991}.\footnote{This behavior is not easily explained by direct barriers, e.g., transaction costs and capital controls \citep*{AHEARNE2004}.  Information asymmetries and lower resulting posterior/residual risk seem to be key drivers of such biases \citep*{AHEARNE2004,veldkamp2009,PORTES2005}.}
	Therefore, homophilous groups of people can end up (rationally) herding on inefficient decisions because those are the ones about which they have the best information.  This then feeds back and becomes persistent.
	Thus, homophily can limit the endogenous {\sl quantity and diversity} of information available to people from their social networks.

	We model this tradeoff by having agents choose between a \textit{safe} action with payoff normalized to 0 and a \textit{risky} action with unknown payoff. For instance, an agent may be choosing between taking a minimum wage job or pursuing a college education.  Agents differ in their costs of taking the action. In the baseline, homophily is on group identity, and groups have different cost distributions, so group homophily induces cost-relevant exposure. In Appendix~\ref{multicost} we also explore homophily on the cost dimension.
	
	Agents know their costs, but do not know the benefit of the risky action.\footnote{This is a somewhat artificial distinction, and one could add uncertainty over costs.  What is important is that there is heterogeneity in overall payoffs and risky-action payoffs are uncertain to at least some agents at an interim stage.}
	Agents with costs above the unknown benefit would be better off taking the safe action, while all those with costs below the benefit should take the risky action.
	
	The model has overlapping generations. 	Current agents observe the choices of their friends from the previous cohort, know their friends' costs, and see whether their friends succeeded or failed (got a net positive or negative payoff) if they took  the risky action.	 Two things limit agents' learning.
	One is that they may have a different cost than their friend(s): for instance an agent with a high cost, upon observing a low-cost agent succeed with the risky action, is not sure whether they will also succeed.
	The other is that they can only learn about the risky action's payoff if some of their friends take the risky action. This enables us to study the competing forces generated by homophily: gaining more information about what one's own experience would be by observing another agent who has a similar cost take the risky action, but potentially failing to learn about the risky action if one does not observe any agents try it. Herding on the safe action by a group of agents can arise as an endogenous impact of homophily.
	
	In particular, agents belong to one of two groups:  blues or greens (e.g., income groups, ethnic groups, genders, caste groups, etc.).   These groups can have different cost distributions.
	Homophily means that greens are more likely to have green friends, and blues are more likely to have blue friends.
	The key implication is that friends are more likely to have similar costs than two people at random.
	This fuels the positive influence of homophily in that someone is more likely to learn from a friend than a randomly selected agent, provided that the friend has tried the risky action. However, in equilibrium blues and greens end up taking the risky action with different probabilities, which means that an agent is more or less likely to have a chance to learn about the risky action depending on which group their friend belongs to.
	
	We find that having fewer friends or facing greater uncertainty about the payoff to the risky action increases the negative impact of homophily (inducing more learning failures for a group).
	In particular, with small numbers of friends and high amounts of homophily, a group can end up herding on inefficient decisions because of a lack of payoff information about the risky action. Once people have more friends or face less uncertainty, then they are more likely to observe the risky action being taken (or make inferences from it not being taken) which improves learning of both groups and overcomes inefficient herding, and eventually it becomes more important to be learning from people with similar costs, and then homophily is beneficial.
	
	These results showcase the importance of using a dynamic equilibrium model to understand the full endogenous impact of homophily on social learning.  Equilibrium actions evolve over time, and the path of that evolution depends on the level of homophily and groups' initial beliefs and costs; ultimately determining what each group eventually learns.  
	
	Our analysis also provides new insights into dynamic policies.  For example, cross-group ties can be especially useful early in the process, when one group  has historically not taken a risky action and can learn from another group that more routinely experiments.  Once both groups are taking the risky action with  sufficient frequency,  homophily provides more insight and can increase the speed and breadth of learning and spread the higher payoff action widely. For instance, an underrepresented group in some endeavor that involves uncertain payoffs may benefit from mentorships from an over-represented group early on in the learning process, but then would benefit more from same-type mentorships after sufficient time, as we detail below.

	There are additional nuances in the results. For instance, perhaps counter-intuitively, higher correlations between the costs across groups increase the positive impact of homophily. The intuition behind this is layered.
	On the one hand, higher correlations in costs across groups lead to greater learning across groups, which would superficially seem to lower the benefit of homophily.  However, this also means that groups tend to act similarly, which then has a feedback of increasing the benefit of homophily. This shows how endogeneity of the action can reverse the impact of homophily.

	Our paper relates to various strands of the social learning literature.  The herding here differs from the usual forms of herding \citep{banerjee1992,bikhchandanihw1992,smiths2000} that we term  \textit{inference herding} and that occur due to cascading inferences of agents who end up choosing actions independently of their private information.   Our agents observe the costs and success or failure of other agents, and so here the herding comes from the fact that some portions of the population do not observe direct payoff information about the risky action, even when it is taken by some other group. Agents can still make indirect action-based inferences from observed safe choices; these inferences generally work through the endogenous action rates rather than through exogenous private-signal cascades.\footnote{One might expect that this relates to the strategic experimentation literature (e.g., \cite{keller2005strategic}). However, our setting only has agents take actions once, and so it is more similar to the herding literature which also involves informational externalities.}  Homophily impacts which groups of the population herd in which ways---there can be simultaneous herds on different actions in equilibrium.   We refer to this type of herding as \textit{sample herding}, to distinguish it from inference herding.
	
	Our model draws on the idea that people learn more from those most similar.  That idea has been explored previously in other contexts \citep{ketss2016,sethi2016communication}, including work that explores the tradeoff between learning from similar versus diverse sources; e.g. \cite{aral2011diversity}.
	Our work provides new insights into how homophily leads to sample herding and into the payoff and inequality implications suggested by such herding.
	
	Our social learning model is ``{\sl active}'' in that information is endogenously generated depending on the choices of the agents.  The effects of homophily have been previously studied in ``{\sl passive}'' settings, in which people communicate exogenously given information through a network.   For example, \cite{golubj2012} show that homophily can slow learning in a model of repeated communication of beliefs and updating of posteriors.  In another study of passive learning,	\cite{lobel2016preferences}  find a tradeoff for homophily.  In contrast to our result, in their setting homophily is less useful in dense networks and more useful in sparse networks. The key difference is the active/passive distinction: we have endogenous actions which generate different information
	as a function of the network, while in \cite{lobel2016preferences} network density affects the ability of agents to interpret the signals, but the signals received do  not change  with network density.
	Thus, having the actions chosen depend upon homophily provides new and different insights, especially regarding homophily-based herding and its implications for efficiency and inequality.
	
	Recent work has shown that homophily in job-market referral networks can lead to differences in behaviors and outcomes across groups \citep{buhai2020social,bolte2020,miller2021dynamics}.
	This can result in inefficiencies and persistent group differences, including inequality and immobility.
	Although some of the consequences are similar, the reasons are quite distinct and
	hence have different policy implications.  In the context of the referral model, reducing homophily is unambiguously good, while here homophily's effects are ambiguous and depend on the density of the network and the underlying cost and information setting.\footnote{One could build a model of referrals in which there are countervailing forces as well, but that goes beyond the existing models.}

	\section{The Model}

	\noindent

	\textbf{Agents:} There is a continuum of agents, normalized to measure one, divided into two groups called \textit{blues} and \textit{greens}, denoted by $\{b, g\}$.
	
	\textbf{Actions and their Values:} Each agent chooses between two actions.
	The first is a \textit{safe} action (e.g., working a minimum wage job) whose payoff is normalized to $0$. The second is a \textit{risky} action (e.g., attending university).
	
	The risky action has a random benefit $v$ that equals $1$ with probability $p$ and $0$ with probability $1-p$. Hence, the expected payoff from the risky action is $p$.
	The risky action's value $v$ is the same across agents and periods.
	
	\textbf{Costs of Actions:} The safe action's cost is normalized to 0.  
	The risky action entails a cost, which depends on the agent's group. Specifically, an agent from group $\theta \in \{b, g\}$ incurs cost $c_\theta > 0$ with probability $\pi_\theta$, and $0$ with probability $1 - \pi_\theta$.
	Consequently, in equilibrium, the risky action is either ex post optimal for every agent or else only for those incurring zero cost.
	
	To avoid ties, we assume that $c_\theta\neq 1$. Thus, the ex post payoff from the risky action to any agent is strictly positive or strictly negative. 
	To fit with the assumption that costs and values are distinct, one could alternatively let the low cost be negative or the lowest value of $v$ be slightly positive, so that certain agents always benefit from the risky action regardless of its payoff.

	\textbf{Information:} Each agent knows their own realized idiosyncratic cost $c$ of taking the risky action, but not the realized value $v$. We next describe how agents form their expectations.

	\subsection*{Learning and Dynamics}
	
	Agents live in overlapping generations. In each period $t\in \{1, 2, \ldots\}$, a new continuum of agents are born and each agent decides which action to take. Before making this decision, agents learn from their friends in the previous generation.
	
	Groups can differ in the number of friends that they have and in the rate of \textit{homophily} in their friendships.
	In particular, the underlying information network is described by a directed network (a directed graphon). Group-\(\theta\) agents, for \(\theta\in\{b,g\}\), have an integer $d_\theta>0$ number of friends from the previous generation. Each of these friends is from the same group with probability $h_\theta$, and from the other group with the remaining probability $1-h_\theta$. The actual distribution of how many of an agent of group $\theta$'s friends are from their own group is then an independent binomial random variable with $d_\theta$ draws, each with probability $h_\theta$.  Each draw of a group picks an agent from that group with uniform probability.

	Agents see each friend's action choice, group, cost, and whether their net payoff was positive or negative, but not the precise value $v$. In particular, each of an agent's $d_\theta$ observations is in the form $(o,c,\theta')\in\{-,+,\emptyset\}\times\{0,c_b,c_g\}\times\{b,g\}$.
	The first dimension $o$ is a summary of the action choice and outcome: $o=+$ indicates that the observed agent took the risky action and got a nonnegative payoff, $o=-$ indicates that  the agent took the risky action and got a negative payoff, and $o=\emptyset$ indicates that the agent took the safe action.

	Agents update their beliefs about $v$ based on their observed vector of $d_\theta$ signals, $\mathbf{s}$, and their knowledge of the equilibrium structure.
	The information that an agent of generation $t+1$ needs as a basis for updating is the equilibrium fraction of each group $\theta$ agents with cost $c$ taking the risky action in period $t$ as a function of the realization of $v$.  Specifically, there are two types of nontrivial inferences that agents can make:
	
	\begin{itemize}
		
		\item \textbf{Direct Inference:} An agent may observe signals of the form  $(+, c_\theta, \theta)$ or $(-, c_\theta, \theta)$, which can indicate that another high-cost agent tried and succeeded or failed  at the risky action. This allows them to update beliefs about whether $v=1$.  For instance, a $+$ signal from a positive cost agent reveals $v=1$, and similarly a $-$ signal from a positive cost reveals that $v=0$. A $+$ signal from a zero cost agent is not payoff-revealing.
		
		\item \textbf{Indirect Inference:} An agent may also observe signals of the form $(\emptyset, c_\theta, \theta)$, indicating that a high-cost agent chose not to take the risky action. In this case, the observing agent must make an indirect inference about $v$ by using Bayes' rule with their understanding of the (differences in) equilibrium behavior between when $v=1$ and when $v=0$.
	\end{itemize}
	
	We presume that agents with  zero cost follow their dominant strategy, and so always take the risky action.   Thus, they send a positive signal for both values of $v$, and provide no information.
	
	Thus, the equilibrium dynamics are fully characterized by the observations and choices of the green and blue agents who have positive costs.     What information can be inferred from a friend from the other group depends on the relative cost levels.     
	
	The presence of zero cost agents is of interest even though they provide no information, since their presence affects the relative fraction of agents in each group who could potentially be informative. 
	
	In Appendix \ref{multicost}, we examine more than two cost levels.  The basic forces that we outline here persist, but the analysis becomes more complex.   We have designed the model to be the simplest that provides all the moving parts that we described in the introduction.

	\subsection*{Equilibrium Inferences}

	We denote the equilibrium fraction of blue agents with cost $c_b$ and green agents with $c_g$ taking the risky action at period $t$ by $b_t(v)$ and $g_t(v)$, respectively. These are functions of the state (the value of $v$), and have some initial conditions $b_0, g_0$. One can choose these initial conditions based on the agents' prior, but we also solve for the more general case in which these are allowed to vary. The system converges to a steady-state (fixed point) $(g^*(v), b^*(v))$.
	
	An agent's posterior belief in period $t+1$, after observing a vector of signals $\mathbf{s}$ from period $t$ agents, is obtained by Bayes' Rule and denoted by $\mathbb{E}_{t+1}\left[v \mid \mathbf{s}, g_t(\cdot), b_t(\cdot) \right]$.
	
	After observing the vector of signals $\mathbf{s}$ at time $t+1$, an agent from group $\theta$ with cost $c_\theta$ takes the risky action if the posterior expectation $\mathbb{E}_{t+1}\left[v \mid \mathbf{s}, g_t(\cdot), b_t(\cdot) \right]\geq c_\theta$.  Otherwise, the agent takes the safe action. We break ties in favor of the risky action, but any rule can be used with corresponding adjustments in the expressions.
	
	The fraction of agents taking the risky action in period $t+1$ of group $\theta$ and cost $c$ corresponds to the probability of receiving a signal profile $\mathbf{s}$ such that $\mathbb{E}_{t+1}\left[v \mid \mathbf{s}, g_t(\cdot), b_t(\cdot) \right]\geq c$.
	Thus, the dynamics are given by:
	\begin{align}
		g_{t+1}(v) &= \mathbb{P} \left\{ \left. \mathbb{E}_{t+1}\left[v \mid \mathbf{s}, g_{t}(\cdot), b_{t}(\cdot) \right] \geq c_g \  \right\vert \  v, g_{t}(v), b_{t}(v)  \right\}. \label{dynamics-1}\\
		b_{t+1}(v) &= \mathbb{P} \left\{ \left. \mathbb{E}_{t+1}\left[v \mid \mathbf{s}, g_{t}(\cdot), b_{t}(\cdot) \right]\geq c_b \  \right\vert \  v, g_{t}(v), b_{t}(v)  \right\}. \label{dynamics-2}
	\end{align}
	
	A vector $(g^*(v),b^*(v)) \in [0,1]^2$ is a {\sl steady state} if it solves the system of equations given by \eqref{dynamics-1} and \eqref{dynamics-2}  with $(g_t,b_t)=(g_{t+1},b_{t+1})=(g,b)$. 
	
	A steady state exists, and we omit the existence proof that follows from a standard fixed-point argument.  There can exist multiple steady states---for instance if none of the high cost agents invest then nothing is learned and for some priors this becomes self-inducing.  If instead agents expect high investment and the state turns out to be good there can exist steady states that involve investment.  Hence, in what follows we indicate when we are referring to specific steady-states, otherwise the properties are generic.

	We now characterize how homophily shapes steady states and can either benefit or hinder learning in the long run.
	
	\section{Countervailing Effects of Homophily}
	
	We begin by noting a key feature of the model: the dynamics are monotonic in the value of the risky action.
	
	\begin{lemma} \label{monotone}
		For every $t \geq 1$,  $b_t(1) \geq b_t(0) \text{ and } g_t(1) \geq g_t(0)$.
	\end{lemma}

	There are two different possibilities for the types of agents that an agent sees and the relative probabilities of seeing these two possibilities are unchanged across states:  
	(i) only zero cost agents,  (ii) some positive costs agents. 
	All that changes across states is which actions some agents take.  The zero cost agents do not affect beliefs
	and so in case (i) the updating is identical regardless of the state.   
	So, we focus on case (ii).  These break into two subcases, one in which none of the positive cost agents take the risky action in either state,  and the other in which at least one of the positive cost agents takes the action in at least one of the states.  In the first subcase there is again no difference in updating and so the same actions are taken regardless of the state.  
	It is thus only in the second subcase that any difference in updating occurs, and then the updating is necessarily different in direction of the actual states since taking the action perfectly reveals the state and so whichever state the action was taken under is learned perfectly and so cannot have a less extreme updating or effect on action than in the other case.

	\subsection{Full Homophily} \label{Full Homophily}
	
	Before analyzing how homophily impacts learning and behavior, it is useful to solve a benchmark case with \textit{full homophily}, $h_b=h_g=1$.
	This is effectively as if there is only one group since agents only ever see their own type. Without loss of generality, we consider the green group.

	When more than one steady state exists, only some are stable. 
	We say that a steady state $(g, b) \in [0,1]^2$ is \textit{stable} if there is an $\varepsilon > 0$ such that for any initial condition $(g_0, b_0)$ within an $\varepsilon$-neighborhood of $(g, b)$, the system's dynamics converge to $(g(v), b(v))$ for almost every $v$.

	\begin{proposition}
		\label{full}
		Suppose that there is full homophily ($h_g=1$) and at least one of $\pi_g$ or $d_g$ is not equal to 1.\footnote{If $\pi_g=d_g=1$, then every $g(1)\in [0,1]$ is a steady state.} Then the steady states are:
		\begin{itemize}
			\item If $c_g \leq p$, then  $\left[ g(0), g(1) \right] =\left[  \left( 1- \pi_g \right)^{d_g}, 1 \right]$ is the unique steady state.
			\item If $c_g>p$ and $\pi_g d_g \leq 1$, then $\left[ g(0), g(1) \right] =\left[  0, 0 \right]$ is the unique steady state.
			\item If $c_g>p$ and $\pi_g d_g > 1$, then there are two steady states: $\left[ g(0), g(1) \right] =\left[  0, 0 \right]$  and $\left[ g(0), g(1) \right] =\left[  0, g^* \right]$ for some $g^* \in (0,1)$. Moreover,  the latter is the unique stable steady state.
		\end{itemize}	
	\end{proposition}

	The logic behind Proposition \ref{full} is sketched as follows.
	When $c_g \leq p$, the default action under the prior is the risky one and so the payoff is learned when the value is 1, and $g(1)=1$.   If $v=0$, agents who see another high cost type either observe a negative payoff (direct inference) or the safe action taken (indirect inference) - both of which reveal that the state is $0$.  Thus, the only agents who take the risky action are those who don't see another high cost agent, and follow their prior.  That happens with probability $g(0) =  \left( 1- \pi_g \right)^{d_g}$.
	
	When $c_g> p$, high cost agents only take the risky action if they see information that causes them to update positively, which can only come from a high cost agent taking the action and getting a positive payoff.  This implies that $g(0)=0$, and also implies that $g(1)=0$ is always a steady state.
	The possibility of $g(1)>0$ requires high enough probability of observing other high cost agents.  In that case, if there is some small $\varepsilon$ of high cost of types who take the risky action, then others learn from that. Given that $\pi_g d_g>1$, this converges upward as on average more see those.  It does not converge to 1, as there are also some who do not observe any high cost types, or observe those from the previous generation who did not take the risky action.

	\subsection{Partial Homophily}
	
	The implications of homophily and cross-group inferences 
	then become evident when there is interaction between the different groups, so we now focus on the case in which $h_\theta \in (0,1)$.

	The cases in which both groups have the same default action are similar to the
	full-homophily benchmark, and so we briefly mention them here and then
	provide details in an appendix.  If \(c_b,c_g\leq p\), positive cost agents in both
	groups take the risky action under the prior. If \(p<c_b,c_g<1\), positive cost
	agents in both groups take the safe action under the prior, but a success by a
	positive cost agent reveals that \(v=1\).\footnote{If \(c_\theta>1\), then a
		positive cost \(\theta\)-agent never earns a positive net payoff from the risky
		action, even when \(v=1\). Such a type therefore cannot generate a success
		signal in the binary baseline. The same-default safe case is consequently
		stated for \(p<c_\theta<1\).}  In Appendix~\ref{app:same-default}, we provide
	corresponding transition equations and steady-state results.

	The interesting case is such that the two groups have different default actions. Without loss of generality we now focus on the case in which $c_g> p\geq c_b>0$. Thus, $c_g$ green agents have the safe action as their default action, and only change to risky action if  they receive information causing them to update to a posterior belief with sufficiently high probability on $v=1$, while cost $c_b$ agents take the risky action as their default action without any updating on the state.
	The latter fact helps learning.	
	It also means that the steady-states are interior in all cases in the sense that there are always some agents taking each action, as there are agents who have priors who lead them to take the risky action with no information and others who do not, and there is always a chance that an agent does not see any high cost agents and so does not update.

	Steady-states are nuanced, however, as agents not only update when they see the payoff to another high cost agent, but also when they observe a high cost agent who
	does not take the risky action.  Those proportions depend on the agent's type and cost and the state, and so there are (nonlinear) interdependencies.

	Although steady states cannot always be fully expressed in closed form, we can deduce comparative statics.
	A particularly interesting comparative static captures the dual nature of homophily.
	Greens are better off seeing the high cost-types who are frequently taking the risky action,
	which could be either greens or blues depending on the context, and depends on the level of homophily in a way that we can characterize.
	
	A steady state $(g^*, b^*)$ is said to be \textit{regular} if it is stable and $g_{t+1}(v)$ and $b_{t+1}(v)$ are differentiable with respect to $g_t(v)$ and $b_t(v)$ in a neighborhood around $(g^*(v), b^*(v))$ for almost every $v$.

	\begin{proposition} \label{prop-simplified-general}
		Let $(g^*(\cdot),b^*(\cdot))$ be a regular steady state. Then $g^*(1)$ is increasing in $h_g$ if and only if $\pi_g g^*(1) > \pi_b b^*(1)$.
	\end{proposition}
	
	Thus, homophily either enhances or impedes learning depending on the equilibrium structure. The condition $\pi_g g^*(1) > \pi_b b^*(1)$ implies that the green high cost agents who take the risky action are more plentiful than the corresponding blue agents, and thus connections to green agents are more informative than blue agents.
	Proposition \ref{prop-simplified-general} follows from an application of the implicit function theorem for comparative statics.
	
	The proposition is not in terms of primitives, but in terms of equilibrium parameters.
	In order to derive comparative statics in terms of primitive parameters, we focus on a specific case. Throughout the rest of the section, we focus on stable equilibria with $b^*(1)=1$.
	This holds true for any $c_b$ below a certain threshold that guarantees the cost is low enough to prevent any indirect inference from convincing the blue group to take the safe action in the good state.  We also maintain the assumption that $c_g>p$, which means that greens only take the risky action if they
	see some evidence of the good state.
	
	In such settings, green agents never take the risky action in state $v=0$, and only take the risky action in state $v=1$ when they see some high-cost agent taking the risky action.
	Blues all take the risky action in the state $v=1$. When $v=0$, blue agents take the risky action if they do not see any blue agents with high cost, as a high cost blue friend who fails directly reveals $v=0$, and under $b^*(1)=1$ a high-cost blue friend taking the safe action also reveals $v=0$.\footnote{Given $b(1)=1$, blues seeing high cost greens taking the safe action still prefer to take the risky action.}
	The dynamics are thus:
	\begin{align*}
		g_t(0)&=0,  &  g_t(1)&=1 - \bigg( 1- \big[h_g \pi_g g_{t-1}(1) + (1-h_g) \pi_b \big] \bigg)^{d_g},\\
		b_t(0)&= \bigg(1-h_b \pi_b \bigg)^{d_b}, & 	b_t(1)&=1.
	\end{align*}	
	These dynamics are decoupled for different values $v$, as indirect inferences are muted. 
	Greens learn from both greens and blues, but blues are always taking the risky action in the high value state, and so that is a corner solution
	and not interacting with green behavior, so $g_t(1)$ depends only on $g_{t-1}(1)$.
	Blues only learn from blues as under $b^*(1)=1$ they only modify their behavior in the low-value state and when seeing someone take the risky action and failing, and greens are not taking the risky action.
	
	These equations yield intuitive comparative statics in terms of primitives.  For these results we examine the interior case in which $\pi_\theta\in (0,1)$.  This simplifies the statements as it rules out corner cases in which the inequalities are no longer strict, but the corner cases are straightforward to calculate.
	
	\begin{proposition} \label{lowblue}
		Consider a setting in which blue agents strictly prefer to take the risky action unless they observe a negative payoff to some agent. Then, $g^*(1)$ is increasing in $h_g$ if and only if $\pi_g>\pi_b$  and $d_g > \bar{d} = \log_{1-\pi_b}\left[\frac{\pi_g-\pi_b}{\pi_g}\right]$.	
		\medskip
		\noindent
		Moreover,  given any $g_{t-1}(1)$:
		\begin{itemize}
			\item $b_t(0)$ (and $b(0)$) is decreasing in $d_b$ and in $\pi_b h_b$,
			\item $g_t(1)$ (and $g(1)$) is increasing in $\pi_g$, $\pi_b$ and $d_g$.\footnote{
				If $g_{t-1}(1)=0$ then $g_t$ is independent of $\pi_g$.	}
		\end{itemize}
	\end{proposition}
	
	The statements about $\pi_\theta$ and $d_\theta$ follow from the facts that
	having more observations (higher $d_\theta$) and more high cost types (higher $\pi_\theta$) leads to greater information.
	
	The remaining statements are about the interaction of homophily.   In this setting, blues are not influenced in their behaviors from observing greens, and so that accounts for $b_t(0)$ being decreasing in $h_b$.
	The comparative static about $h_g$ shows that homophily can be either beneficial or harmful. We illustrate  this relationship and show how the effect of degree is greater as homophily is increased in figure  \ref{fig:crossing}.
	
	\begin{figure}[h!!]
		\centering
		\includegraphics[width=1\textwidth]{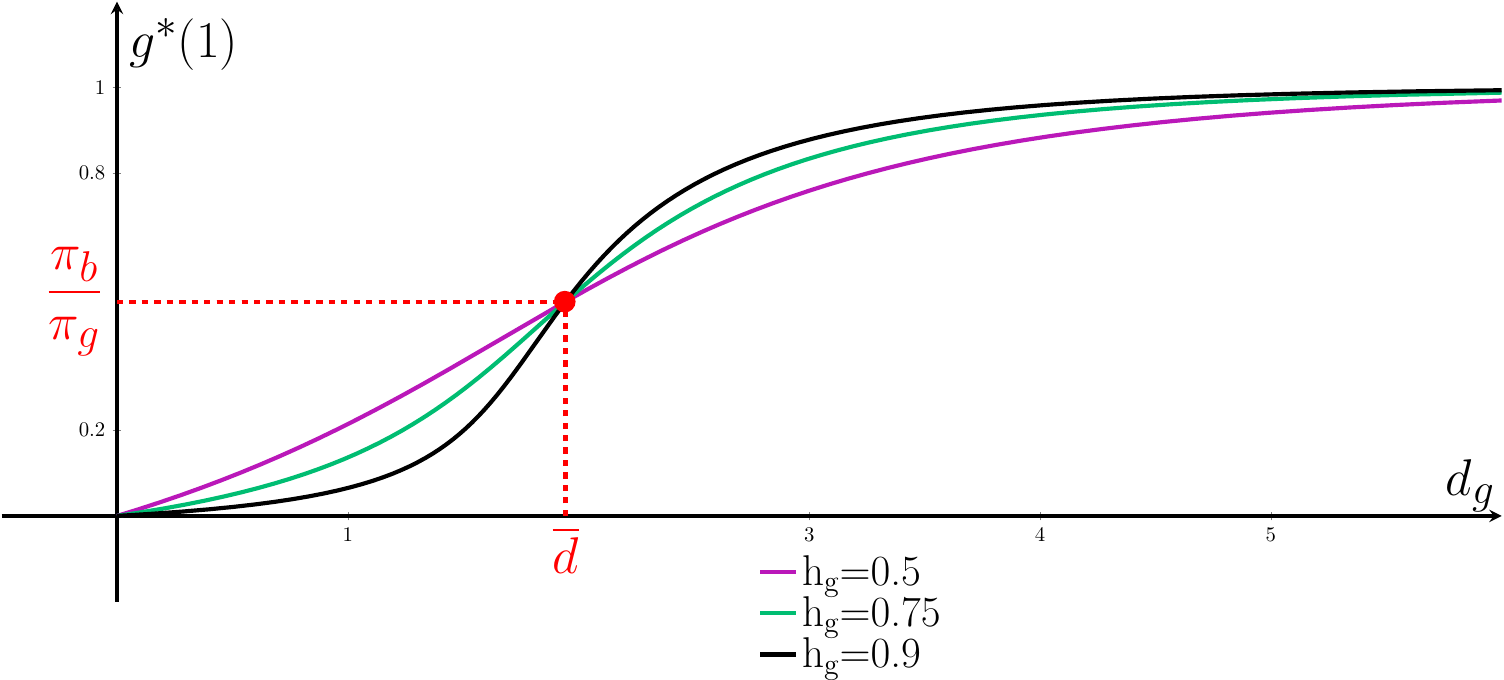}
		\caption[Caption of Crossing]{Steady state level of green group agents taking the risky action when $\pi_g=0.6, \pi_b=0.3$. We show the plot for real values of $d_g$ but the actual values of $d_g$ are discrete. }
		\label{fig:crossing}
	\end{figure}
	
	In particular, greens can learn both from blues and greens who take the risky action, and what is consequential is the relative fraction of each who are taking the risky action.
	That depends on the degrees, which accounts for the condition $d_g >  \bar{d}$.  For high enough degree, greens get many observations and so are likely to take the risky action, and so that favors learning from greens, to the extent that $\pi_g > \pi_b$.
	If instead $\pi_g \leq  \pi_b$, then seeing blues dominates seeing greens.
	In this case, green agents are not taking the risky action often, and increased homophily among greens reduces their information about the risky action leading them to herd to the safe action. This is the {sample herding} we referred to in the introduction.
	
	\subsection{More Levels of Costs and Values}
	
	Appendix~\ref{multicost} extends the analysis to finitely many cost levels and values of the risky action, allowing friendship patterns to depend on costs in addition to group identity (in the base model these are effectively the same thing). This richer environment delivers two new implications.
	
	First, complete learning---agents behaving as if they knew $v$ and taking the risky action if and only if $v$ exceeds their cost---is the unique stable steady state when the network exhibits perfect cost homophily, so that observed experimentation outcomes are informative for an observer's own cutoff.
	
	Second, the extension shows how homophily measured along one attribute can be generated by sorting on another. In particular, homophily in blue/green space can arise even when friendship formation is independent of blue/green identity conditional on costs. If agents match assortatively on costs and the blue and green groups have different cost distributions, then blues end up disproportionately linked to blues and greens to greens. Thus, cost assortativity generates \textit{incidental homophily} in blue/green space, despite the absence of any direct preference for same-color friendships. Moreover, when one group’s cost distribution likelihood-ratio dominates the other’s, the induced blue/green homophily varies monotonically with cost, so homophily is itself assortative, with a threshold cost at which the direction of homophily/heterophily flips across groups.

	\section{Concluding Remarks}
	
	The benefits of group/cost homophily arise in the long run after a group has begun to learn the state sufficiently and takes informative actions, and otherwise it can be inhibiting, along the lines of Proposition \ref{lowblue}.
	This suggests policies that encourage cross-group relationships when there are substantial differences in risky-action choices by groups (e.g., consistent with \cite{chettyetal2022I}),
	but then allowing cost-homophily once learning from own types becomes sufficient.  
	
	In terms of policy implications that one gets within our model, there are two clear ones. First, policies disseminating information about success rates by cost can be much more enlightening than simply providing success rates. Second, subsidizing (or insuring) risky choice by a group can reverse a sample herding problem that arises because of homophily.   Once sufficiently seeded, that policy might only need be implemented for a short period, but could have a lasting impact.

	\newpage
	\begin{spacing}{0.8}
		\bibliographystyle{jpe}
		\bibliography{InequalityNetworks}
	\end{spacing}
	
	\begin{appendices}
		
		\section{Proofs}
		
		\noindent
		\textbf{Proof of Lemma \ref{monotone}:}
		We prove the result for green agents. The proof for blue agents is identical.
		
		Let \(S\) be the set of signal profiles with no direct payoff signal from a positive cost agent. Thus profiles in \(S\) contain only zero cost \(+\) signals and safe-action signals of the form \((\emptyset,c_\theta,\theta)\), \(\theta\in\{b,g\}\). 
		
		For each \(\mathbf{s}\in S\), define $\mu(\mathbf{s}) = \mathbb{P}\!\left(\mathbf{s}\mid v=1,g_t(1),b_t(1)\right),$ and $\nu(\mathbf{s}) = \mathbb{P}\!\left(\mathbf{s}\mid v=0,g_t(0),b_t(0)\right).$
		
		Let \(R\subseteq S\) be the set of profiles after which a positive cost green agent chooses the risky action: $R= \left\{ \mathbf{s}\in S: \mathbb{E}_{t+1}\!\left[v\mid \mathbf{s},g_t(\cdot),b_t(\cdot)\right]\geq c_g \right\}.$
		
		Profiles with zero probability under both states can be ignored, and for any profile with positive probability, by Bayes' rule, we have that:
		\[
		\mathbf{s}\in R
		\quad\Longleftrightarrow\quad
		\mu(\mathbf{s})
		\geq
		\underbrace{\frac{c_g(1-p)}{p(1-c_g)}}_{:=\kappa}\nu(\mathbf{s}).
		\]

		Under \(v=0\), any direct payoff signal from a positive cost agent is negative and reveals that the state is bad, so it cannot induce the risky action. Therefore,
		\[
		g_{t+1}(0)=\sum_{\mathbf{s}\in R}\nu(\mathbf{s}).
		\]
		Under \(v=1\), any direct payoff signal from a positive cost agent is positive and reveals that the state is good. Since \(c_g<1\), any such signal induces the risky action. Hence
		\[
		g_{t+1}(1)
		=
		1-\sum_{\mathbf{s}\in S}\mu(\mathbf{s})
		+
		\sum_{\mathbf{s}\in R}\mu(\mathbf{s}).
		\]
		
		If \(\kappa\geq 1\), then for every \(\mathbf{s}\in R\), $\mu(\mathbf{s})\geq \kappa\nu(\mathbf{s})\geq \nu(\mathbf{s}),$ so
		\[ 
		g_{t+1}(1) \geq \sum_{\mathbf{s}\in R}\mu(\mathbf{s}) \geq \sum_{\mathbf{s}\in R}\nu(\mathbf{s}) = g_{t+1}(0).
		\]
		
		If \(\kappa<1\), then for every \(\mathbf{s}\in S\setminus R\), $\mu(\mathbf{s})<\kappa\nu(\mathbf{s})<\nu(\mathbf{s}).$ Thus
		\[
		\sum_{\mathbf{s}\in R}\left[\nu(\mathbf{s})-\mu(\mathbf{s})\right]
		\leq
		\sum_{\mathbf{s}\in S}\left[\nu(\mathbf{s})-\mu(\mathbf{s})\right].
		\]
		Therefore,
		\[
		\begin{aligned}
			g_{t+1}(1)-g_{t+1}(0)
			&=
			1-\sum_{\mathbf{s}\in S}\mu(\mathbf{s})
			-
			\sum_{\mathbf{s}\in R}\left[\nu(\mathbf{s})-\mu(\mathbf{s})\right] \\
			&\geq
			1-\sum_{\mathbf{s}\in S}\mu(\mathbf{s})
			-
			\sum_{\mathbf{s}\in S}\left[\nu(\mathbf{s})-\mu(\mathbf{s})\right] \\
			&=
			1-\sum_{\mathbf{s}\in S}\nu(\mathbf{s})
			\geq 0.
		\end{aligned}
		\]
		Thus \(g_{t+1}(1)\geq g_{t+1}(0)\).\eproof

		\bigskip
		\noindent
		\textbf{Proof of Proposition \ref{full}:} Suppose that $c_g \leq p$. It can be verified that $g(1)=1$ and $ g(0)=  \left( 1- \pi_g \right)^{d_g}$ is a steady state.
		
		We define the following set of signals, $ S_+ = \{ \mathbf{s}_g : (+,c_g,g) \in \mathbf{s}_g \} \cup \{ (+,0,g)^{d_g} \} $.  Any green agent receiving $\mathbf{s}_g \in S_+$ has a posterior belief above $c_g$ and takes the  risky action. Thus, we have that:
		\begin{align*}
			g_{t+1}(1)  \geq \mathbb{P}_g( \mathbf{s}_g \in S_+ \mid v=1, g_t(1))   = 1 - (1-\pi_g g_t(1))^{d_g} + (1-\pi_g)^{d_g}.
		\end{align*}
		Every steady state $g_t(1)=g_{t+1}(1)=g^*$  solves this inequality, and for $g(1)=1$ it holds with equality. Observe that $\mathbb{P}_g( \mathbf{s}_g \in S \vert v=1, g) $ is  concave  in $g$ whenever $d_g \geq 1$. Since $\Gamma(0)>0$, we conclude $\Gamma(g)>g $ for any $g \in [0,1)$. Thus, $g(1)=1$ is the unique steady state.
		
		Now consider the case where $c_g>p$. We start by showing that every steady state has $g(0)=0$. Suppose for the contradiction $g(0)>0$. Thus, there exists a signal profile $\mathbf{s}_g$ such that $\mathbb{E}[v \vert \mathbf{s}_g, g_t(0)] \geq c_g>p$, and $\mathbb{P}( \mathbf{s}_g \vert 0, g_t(0))>0$. The latter also requires that $\mathbf{s}_g$ consists of signals $(\emptyset,c_g,g)$ and $(+,0,g)$. However, by Bayes' rule, such a signal vector  $\mathbf{s}_g$ exists if and only if $g(0)>g(1)$. This contradicts the conclusion of Lemma 1.
		
		Note that, when $c_g>p$, $g_{t+1}(1)$ equals the probability of hearing at least one positive signal from cost $c_g$ types since indirect inferences only reduce the posterior. This probability is given by $\Gamma(g_t(1)) = 1- (1-\pi_g g_t(1) )^{d_g}$. 	
		
		Observe that $g(1)=0$ is a solution $g=\Gamma(g)$. Moreover, $\Gamma(g)$ is concave, continuous in $g$ and $\Gamma'(0) = \pi_g d_g$.	If $\Gamma'(0) \leq 1$, by concavity, $\Gamma(g)<g$ for every $g \in (0,1]$. Thus $g(1)=0$ is the unique solution to $\Gamma(g)=g$. However, if $\Gamma'(0) > 1$, then $g(1)=0$ is not a stable steady state. It follows from $\Gamma(1)<1$ there exists a steady state $g \in (0,1)$. By concavity this steady state is stable. \eproof
		
		\bigskip
		\noindent \textbf{Proof of Proposition \ref{prop-simplified-general}:} We define $\alpha_t(\cdot)=\begin{bmatrix} g_t(1) &  b_t(1)  & b_t(0)  & g_t(0) \end{bmatrix}^T$. The dynamics are given by $\Gamma (\alpha_t)= \begin{bmatrix} g_{t+1}(1 ) & b_{t+1}(1) & b_{t+1}(0) & g_t(0) \end{bmatrix}^T$. In Appendix \ref{app:detail}, we
		describe the details and closed form equations for the dynamics. Moreover, we show that $g_{t+1}(0)=0$ for any $\alpha_t$. So the dynamics can be studied in the reduced state space $\alpha_t(\cdot)=\begin{bmatrix} g_t(1) &  b_t(1)  & b_t(0) \end{bmatrix}^T$

		Let $\alpha^*$ be a regular steady state for $\Gamma(\alpha)$.  By the Hartman-Grobman Theorem, the qualitative properties of a steady state can be analyzed using the linearized system $J_{\Gamma}(\alpha)=\alpha$ where $J_{\Gamma}$ is the Jacobian of $\Gamma(\alpha)$, which exists since $\alpha^*$ is a regular steady state.
		
		In Appendix \ref{app:detail} we  also show that if the dynamics are differentiable around $\alpha^*$, then $b_{t+1}(0)$ and $g_{t+1}(0)$ does not explicitly depend on $b_{t+1}(1)$ and $g_{t+1}(1)$, and vice versa.\footnote{Alternatively note that $J_{\Gamma}$ is non-negative and its spectral radius is $<1$ at a stable fixed point, the inverse of $I-J_{\Gamma}$ exists and equals $\sum_{k \geq 0} J_{\Gamma}^k$, which is entry-wise positive.} We conclude that locally linearized dynamics for $b_{t+1}(1)$ and $g_{t+1}(1)$ only depends on $b_{t}(1)$ and $g_{t}(1)$. So the dynamics for $v=0$ and $v=1$ are decoupled.
		
		We define the Jacobian $J_{\Gamma}$ for the decoupled dynamics for $v=1$ as:
		$$
		J_{\Gamma}=\left[\begin{array}{ll}
			\frac{\partial g_{t+1}(1)}{\partial g_t(1)} & \frac{\partial g_{t+1}(1)}{\partial b_t(1)} \\
			\frac{\partial b_{t+1}(1)}{\partial g_t(1)} & \frac{\partial b_{t+1}(1)}{\partial b_t(1)}
		\end{array}\right].
		$$
		
		Equations \eqref{green-partials} \eqref{blue-partial-1} and \eqref{blue-partial-2} show that $J_\Gamma$ is a positive matrix.  Moreover, in a regular (hence stable) steady state $\alpha^*$, the Jacobian $J_{\Gamma}$ has spectral radius smaller than 1. So, we conclude that $I-J_{\Gamma}$ is an $M$-matrix, when evaluated at a regular steady state.\footnote{It has negative off-diagonal entries and has eigenvalues whose real parts are non-negative.} $M$-matrices are invertible, and their inverse is a positive matrix. We conclude $\left(I-J_{\Gamma}\right)^{-1}$ is a positive matrix.

		Applying the Implicit Function Theorem to the $v=1$ subsystem, we obtain:
		$$
		\begin{bmatrix}
			\frac{\partial g^*(1)}{\partial h_g} & \frac{\partial b^*(1)}{\partial h_g}
		\end{bmatrix}^T=\left(I-J_{\Gamma}\right)^{-1} \begin{bmatrix}
			\frac{\partial g_{t+1}(1)}{\partial h_g} & \frac{\partial b_{t+1}(1)}{\partial h_g}
		\end{bmatrix}^T =  \left(I-J_{\Gamma}\right)^{-1} \begin{bmatrix}
			\frac{\partial g_{t+1}(1)}{\partial h_g} & 0
		\end{bmatrix}^T.$$
		
		The second equality follows from the immediate observation $\frac{\partial b_{t+1}(1)}{\partial h_g}=0$, as shown in equation \eqref{partial-hg} in Appendix~C. Multiplication by the first column of $(I-J_\Gamma)^{-1}$ therefore preserves the sign of the first entry of the right-hand-side vector. Thus,
		$$
		\text{sign} \left( \frac{\partial g^*(1)}{\partial h_g} \right) = \text{sign} \left( \frac{\partial g_{t+1}(1)}{\partial h_g} \right) =  \text{sign} \left(  \pi_g g^*(1) - \pi_b b^*(1) \right).
		$$
		where the second equality is again shown in equation \eqref{partial-hg}.  \eproof

		\bigskip
		\noindent {\bf Proof of Proposition \ref{lowblue}:}
		When \(b^*(1)=1\), the steady state \(g^*(1)\) satisfies the fixed‐point equation
		\[
		g = \Gamma(g) \quad \text{with} \quad \Gamma(g)= 1 - \Bigl(1 - \bigl[h_g \pi_g g + (1-h_g) \pi_b\bigr]\Bigr)^{d_g}.
		\]
		Since \(
		\Gamma(0)=1-\Bigl(1-(1-h_g)\pi_b\Bigr)^{d_g}>0 \), \( \Gamma(1)=1-\Bigl(1-h_g\pi_g-(1-h_g)\pi_b\Bigr)^{d_g}<1,\) and because \(\Gamma(g)\) is continuous and concave on \([0,1]\), there exists a unique interior fixed point \(g^*(1) \in (0,1)\).
		
		Next, since $\Gamma(g)$ is increasing in $g$, the comparative statics follow from the partial derivatives of \(\Gamma(g)\) by the implicit function theorem. In particular, it follows that  \(g^*(1)\) is increasing in \(h_g\) if and only if
		\[
		\frac{\partial \Gamma(g)}{\partial h_g}  = A (\pi_g\,g^*(1) - \pi_b) >0,
		\]
		where $A$ is positive. We conclude \(g^*(1)\) is increasing in \(h_g\) if and only if $\pi_g g^*(1) > \pi_b$. Similarly, one can show that \(g^*(1)\) is strictly increasing in \(d_g\).
		
		For the remainder of the proof, suppose that \(\pi_g > \pi_b\) (the case \(\pi_g \leq \pi_b\) is trivial since then \(\pi_g\,g^*(1) \leq \pi_b\) for every \(g^*(1) \in [0,1]\), implying that the steady state is decreasing in \(h_g\)). Taking the natural logarithm of the fixed-point equation \(g^*(1)=\Gamma(g^*(1))\) and rearranging yields
		\[
		d_g=\frac{\ln\Bigl(1-g^*(1)\Bigr)}{\ln\Bigl(1-h_g\pi_g g^*(1)-(1-h_g)\pi_b\Bigr)}.
		\]
		By the Inverse Function Theorem and the fact that \(g^*(1)\) increases in \(d_g\), the right-hand side is an increasing function of \(g^*(1)\). Now, define \[ \bar{g}=\frac{\pi_b}{\pi_g} \in [0,1], \quad  \text{so we have that} \quad
		1-\bar{g}=1-\frac{\pi_b}{\pi_g} \quad \text{and} \quad 1-h_g\pi_g\frac{\pi_b}{\pi_g}-(1-h_g)\pi_b=1-\pi_b.
		\]
		The equation for \(d_g\) at $g^*(1)=\bar{g}$ simplifies to
		\[
		d_g=\frac{\ln\Bigl(1-\frac{\pi_b}{\pi_g}\Bigr)}{\ln(1-\pi_b)}=\log_{1-\pi_b}\frac{\pi_g-\pi_b}{\pi_g}.
		\]
		By the monotonicity of the right-hand side in \(g^*(1)\), it follows that
		\[
		d_g>\log_{1-\pi_b}\frac{\pi_g-\pi_b}{\pi_g} \quad \text{if and only if} \quad g^*(1)>\bar{g}.
		\]
		Recalling that \(g^*(1)\) is increasing in \(h_g\) if and only if \(\pi_g\,g^*(1) > \pi_b\), we conclude that, under the assumption  \(\pi_g> \pi_b\), an increase in \(h_g\) raises the steady state \(g^*(1)\) precisely when
		\(
		d_g>\log_{1-\pi_b}\frac{\pi_g-\pi_b}{\pi_g}
		\).
		
		Finally the claim for $b_t(0)$ follows immediately from $b_t(0)= \bigg(1-h_b \pi_b \bigg)^{d_b}$, and the remaining monotonicity claims for \(g_t(1)\) follow directly from $g_t(1)=1-\left(1-h_g\pi_g g_{t-1}(1)-(1-h_g)\pi_b\right)^{d_g}.$ \eproof
		
		\bigskip

		\section{Multiple Costs and Values, Assortativity in Costs, and Incidental Homophily}\label{multicost}
		
		We extend the model in the main text to multiple values and costs. 
		
		The risky-action payoff is $v-c$, where the value
		$v$ is drawn from a finite support $\mathcal V\subset \mathbb R_+$ with prior $\Pr(v)$,
		and costs $c$ are drawn from a finite support $\mathcal C\subset \mathbb R_+$ with
		group-dependent distributions $\Pr_\theta(c)$.
		
		As in the main text, agents observe from their friends only the sign of the realized payoff from the risky action (success/failure) and whether the risky action was taken. Moreover, as in the body of the paper, costs and values are assumed to satisfy $\mathcal V\cap \mathcal C=\varnothing$ to avoid indifference.
		
		Instead of just tracking homophily in blue-green types, we also allow homophily to depend on costs. Thus, we now track an agent's type as a pair $(\theta, c) \in\{b, g\} \times \mathcal C$. An agent with type $(\theta,c)$ draws each of her $d_{\theta,c}\in\{1,2,\ldots\}$ friends from the previous generation independently from a distribution
		$h_{\theta,c}(\theta',c')$ over $\{b,g\}\times\mathcal C$.
		
		Let \( \alpha_t(\theta,c,v)\in[0,1] \) denote the period-$t$ fraction of type $(\theta,c)$ agents taking the risky action
		when the realized value is $v$. Write $g_t(c,v):=\alpha_t(g,c,v)$ and $b_t(c,v):=\alpha_t(b,c,v)$. A \textit{steady state} is a collection
		$\alpha^*(\theta,c,v)$ that reproduces itself under Bayesian updating and best responses,
		analogous to equations \eqref{dynamics-1}--\eqref{dynamics-2} in the main text.
		
		Throughout this appendix we impose the following mild richness condition:
		\begin{assumption}\label{ass:nontrivial}
			For every cost $c\in\mathcal C$, there exist $v^-,v^+\in\mathcal V$ with $v^-<c<v^+$.
		\end{assumption}
		
		Assumption \ref{ass:nontrivial} ensures that for every cost level, the risky action is ex-post optimal in some states and suboptimal in others.
		
		\subsection{Perfect cost homophily and complete learning}\label{app:perfect-cost-homophily}
		If the costs of two agents are
		too different from each other then one learns little from whether the other succeeds or fails at the risky action. Agents learn the most from observing others whose costs are ``close enough'' that success/failure has the same implication for whether $v$ lies above or below their own cost.
		
		We say the network exhibits \textit{perfect cost homophily} if whenever a type
		$(\theta,c)$ can observe a type $(\theta',c')$ with positive probability, there is no
		value strictly between their costs. 
		
		More formally, there is perfect cost homophily if for every $(\theta,c)$ and $(\theta',c')$,
		\[
		h_{\theta,c}(\theta',c')>0
		\quad\Longrightarrow\quad
		\not\exists v\in\mathcal V \text{ such that } \min\{c,c'\}<v<\max\{c,c'\}.
		\]

		We also say there is \textit{complete learning} at a steady state if agents behave as if they knew $v$. That is, a steady state $\alpha^*$ exhibits complete learning if $\alpha^*(\theta,c,v)=1$ whenever $c<v$ and $\alpha^*(\theta,c,v)=0$ whenever $c>v$.
		
		Proposition \ref{prop:cl-iff-pch} shows that these two conditions are equivalent.
		
		\begin{proposition} \label{prop:cl-iff-pch}
			Suppose $d_{\theta,c} > 1$ for all $(\theta,c)$ and Assumption \ref{ass:nontrivial} holds.
			Then complete learning is the unique stable steady state if and only if the network exhibits
			perfect cost homophily.
		\end{proposition}

		\noindent\textbf{Proof of Proposition \ref{prop:cl-iff-pch}.}
		
		\underline{\textit{If.}} Suppose the network exhibits perfect cost homophily. We first show that complete learning is a steady state. Fix a type \((\theta,c)\) and a value \(v\). If \(c<v\), then every possible friend type \((\theta',c')\) must also satisfy \(c'<v\). Otherwise \(v\) would lie strictly between \(c\) and \(c'\), contradicting perfect cost homophily. Thus, under complete learning, every friend experiments and succeeds. Since no value in \(\mathcal V\) lies between \(c\) and any \(c'\), such a success implies \(v>c\), so type \((\theta,c)\) also takes the risky action.
		
		If \(c>v\), then every friend type \((\theta',c')\) must also satisfy \(c'>v\). Under complete learning, every possible friend takes the safe action. Since no value in \(\mathcal V\) lies between \(c\) and any  \(c'\), these safe choices imply \(v<c\), so type \((\theta,c)\) takes the safe action. Hence we conclude that complete learning reproduces itself and is a steady state.
		
		We next show stability and uniqueness. Let \(\alpha_t(\theta,c,v)\) denote the fraction of type \((\theta,c)\) agents taking the risky action at state $v$. Consider first a type \((\theta,c)\) with \(c<v\). Under perfect cost homophily, any observable friend type has cost on the same side of \(v\) as \(c\). Thus any observed successful experiment by an observable friend type reveals \(v>c\). Defining
		\[
		\rho_{\theta,c,t}(v)
		=
		\sum_{(\theta',c'):\,h_{\theta,c}(\theta',c')>0}
		h_{\theta,c}(\theta',c')\,\alpha_t(\theta',c',v).
		\]
		we can state a lower bound to the fraction of type \((\theta,c)\) agents taking the risky action: $$\alpha_{t+1}(\theta,c,v)
		\geq
		1-\left(1-\rho_{\theta,c,t}(v)\right)^{d_{\theta,c}}.
		$$
		
		In an \(\varepsilon\)-neighborhood of complete learning, \(\alpha_t(\theta',c',v)\geq 1-\varepsilon\) for every observable friend type with \(c'<v\). Hence \(\rho_{\theta,c,t}(v)\geq 1-\varepsilon\), and $\alpha_{t+1}(\theta,c,v)
		\geq
		1-\varepsilon^{d_{\theta,c}}
		>
		1-\varepsilon$ for \(d_{\theta,c}>1\) and \(\varepsilon\) small. Thus the complete-learning value \(\alpha^*(\theta,c,v)=1\) is locally stable.
		
		The same bound $\alpha_{t+1}(\theta,c,v)
		\geq
		1-\varepsilon^{d_{\theta,c}}$ rules out any stable steady state with \(\alpha^*(\theta,c,v)<1\) when \(c<v\). If experimentation by the relevant observable friend types is perturbed upward by a small amount, then the probability of observing at least one successful experiment is amplified when \(d_{\theta,c}>1\). Thus a steady state with \(\alpha^*(\theta,c,v)<1\) cannot be stable. Therefore every stable steady state has \(\alpha^*(\theta,c,v)=1\) whenever \(c<v\).
		
		Now consider a type \((\theta,c)\) with \(c>v\). From the previous paragraph, every stable steady state has \(\alpha^*(\theta',c',v)=1\) for all observable friend types with \(c'<v\). Under perfect cost homophily, however, a type with \(c>v\) cannot observe any type with cost below \(v\). Hence observable friend types have costs above \(v\). If such a friend experiments, the negative payoff reveals that \(v<c\). If such friends take the safe action, then, in a stable steady state, the absence of experimentation by observable friend types whose costs are on the same side of \(v\) also supports the inference that \(v<c\). Thus type \((\theta,c)\) takes the safe action in every stable steady state. Hence \(\alpha^*(\theta,c,v)=0\) whenever \(c>v\).
		
		Combining the cases \(c<v\) and \(c>v\), complete learning is the unique stable steady state.
		
		\medskip
		\underline{\textit{Only if.}} We prove the converse by contrapositive. Suppose the network does not exhibit perfect cost homophily. Then there exist types \((\theta,c)\) and \((\theta',c')\) such that \(h_{\theta,c}(\theta',c')>0\) and some \(v\in\mathcal V\) lies strictly between \(c\) and \(c'\).
		
		First suppose \(c>v>c'\). Under complete learning at state \(v\), type \((\theta',c')\) agents take the risky action and succeed. There is a positive measure of type \((\theta,c)\) agents who observe only such signals. If this signal profile induces the risky action for type \((\theta,c)\), then complete learning fails at state \(v\), since \(c>v\). If it does not induce the risky action, then consider any \(v'\in\mathcal V\) with \(v'>c\), which exists by Assumption \ref{ass:nontrivial}. At state \(v'\), the same signal profile occurs with positive probability and gives the same posterior, so type \((\theta,c)\) still takes the safe action. This contradicts complete learning at state \(v'\), where \(v'>c\).
		
		Now suppose \(c'<v<c\). Under complete learning at state \(v\), type \((\theta',c')\) agents take the risky action and succeed, while type \((\theta,c)\) agents should take the safe action. The same argument as above gives a contradiction.
		
		Finally suppose \(c'>v>c\). Under complete learning at state \(v\), type \((\theta',c')\) agents take the safe action. There is a positive measure of type \((\theta,c)\) agents who observe only such safe choices. If this signal profile induces the safe action for type \((\theta,c)\), then complete learning fails at state \(v\), since \(c<v\). If it induces the risky action, then consider any \(v''\in\mathcal V\) with \(v''<c\), which exists by Assumption \ref{ass:nontrivial}. At state \(v''\), the same signal profile occurs with positive probability and gives the same posterior, so type \((\theta,c)\) still takes the risky action. This contradicts complete learning at state \(v''\), where \(v''<c\). Therefore complete learning cannot be a steady state unless the network exhibits perfect cost homophily. \eproof

		\subsection{Incidental homophily from cost assortativity}\label{app:incidental}
		
		Next, we explore other implications of cost homophily. We show that homophily in one dimension (costs) can lead to incidental homophily on another dimension (blue/green), even when there is no homophily on that second dimension.
		
		For the remainder, we assume blues and greens have the same cost support and Assumption \ref{ass:nontrivial} holds.
		
		We say there is \textit{color-blind perfect cost homophily} if:
		\begin{enumerate}[(i)]
			\item \textbf{Perfect cost homophily.} All observed friends of type $(\theta,c)$ have cost $c$; and
			\item \textbf{Color-blind mixing within cost}. Conditional on cost $c$, friend color is drawn in
			proportion to the population composition at that cost:
			\[
			h_{\theta,c}(\theta',c)
			=\frac{\lambda_{\theta'}\,\Pr_{\theta'}(c)}{\lambda_g\Pr_g(c)+\lambda_b\Pr_b(c)}
			\qquad\text{for all }\theta,\theta',c.
			\]
		\end{enumerate}
		As we show next, this not only implies average homophily in blue/green space, but also implies heterogeneity in that homophily, and that own-type links are assortative in blue/green homophily.

		We say that $\operatorname{Pr}_g(\cdot)$ likelihood ratio dominates $\operatorname{Pr}_b(\cdot)$ if $\frac{\operatorname{Pr}_g(c)}{\operatorname{Pr}_b(c)}$ is increasing in $c$ on the support.

		\begin{proposition} \label{prop:incidental}
			Assume $\Pr_g(\cdot)\neq \Pr_b(\cdot)$ and color-blind perfect cost homophily holds.
			
			\begin{enumerate}[(i)]
				\item \textbf{Average incidental homophily.} For each $\theta\in\{b,g\}$,
				\[
				\sum_{c\in\mathcal C}\Pr_\theta(c)\,h_{\theta,c}(\theta,c) \;>\; \lambda_\theta.
				\]
				
				\item \textbf{Monotonicity of color homophily in cost.}  If
				$\Pr_g(\cdot)$ likelihood-ratio dominates $\Pr_b(\cdot)$, then
				$h_{g,c}(g,c)$ is increasing in $c$ and $h_{b,c}(b,c)$ is decreasing in $c$. Moreover there exists a threshold cost $\bar{c}:=\min \{c \in \mathcal{C}: r(c) \geq 1\}$ that determines homophily/heterophily in opposite ways across groups: $h_{g, c}(g, c)>\lambda_g$ and $h_{b, c}(b, c)<\lambda_b$ whenever $c\geq\bar{c}$, and the inequalities are reversed if $c<\bar{c}$.
			\end{enumerate}
		\end{proposition}
		
		Figure \ref{fig:hom_types} illustrates Proposition \ref{prop:incidental}.
		
		\begin{figure}[h!!]
			\centering
			\includegraphics[width=0.6\textwidth]{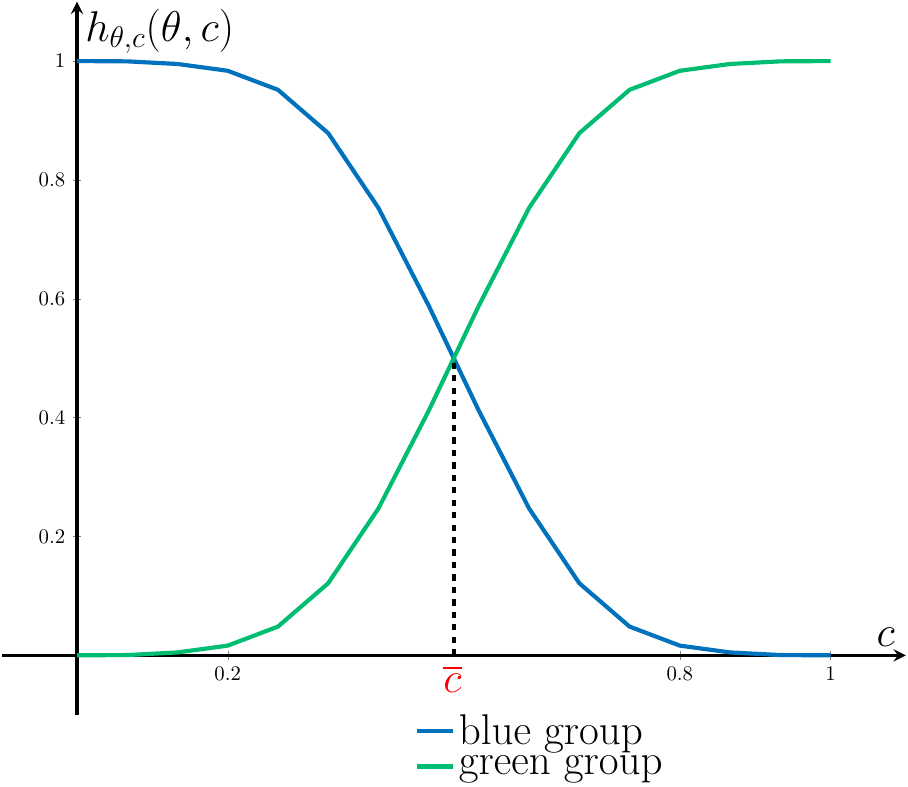}
			\caption{Blue/green homophily as a function of cost when there is perfect cost homophily with no attention to green/blue and $\Pr_g$ likelihood ratio dominates $\Pr_b.$ The figure is for equal-sized groups.}
			\label{fig:hom_types}
		\end{figure}
		
		Proposition \ref{prop:incidental} shows that perfect homophily in costs with no attention to green/blue leads to average homophily in greens and blues. That is, although the network is determined by costs in a color-blind way, in the resulting network the greens are relatively more likely to be linked to greens and blues to blues. This only requires that the cost distribution differ between the two groups.
		
		Moreover, if the groups are ordered in terms of their costs, then homophily is assortative. Higher cost agents from a higher cost group are more homophilous (on blue/green) while lower cost agents from the same group are less homophilous. Finally, green and blue agents have opposite homophily/heterophily above and below a threshold cost $\bar{c}$. This means that cross group links (since they are based on cost) end up being negatively assortative: homophilous greens connect with heterophilous blues and vice versa.
		
		\medskip
		\noindent\textbf{Proof of Proposition \ref{prop:incidental}.}
		Under color-blind perfect cost homophily,
		\[
		h_{\theta,c}(\theta,c)=\frac{\lambda_\theta\,\Pr_\theta(c)}{\lambda_\theta\Pr_\theta(c)
			+(1-\lambda_\theta)\Pr_{\theta'}(c)},
		\qquad \theta'\neq \theta.
		\]
		
		\emph{(i) Average incidental homophily.}
		Compute
		\[
		\sum_{c}\Pr_\theta(c)\,h_{\theta,c}(\theta,c)
		= \lambda_\theta\sum_c \frac{\Pr_\theta(c)^2}{\lambda_\theta\Pr_\theta(c)+(1-\lambda_\theta)\Pr_{\theta'}(c)}.
		\]
		Apply Cauchy--Schwarz to the vectors
		$x_c=\Pr_\theta(c)/\sqrt{\lambda_\theta\Pr_\theta(c)+(1-\lambda_\theta)\Pr_{\theta'}(c)}$
		and
		$y_c=\sqrt{\lambda_\theta\Pr_\theta(c)+(1-\lambda_\theta)\Pr_{\theta'}(c)}$:
		\[
		\Big(\sum_c x_c^2\Big)\Big(\sum_c y_c^2\Big)\;\ge\; \Big(\sum_c x_cy_c\Big)^2
		\;\;\Longrightarrow\;\;
		\sum_c \frac{\Pr_\theta(c)^2}{\lambda_\theta\Pr_\theta(c)+(1-\lambda_\theta)\Pr_{\theta'}(c)}
		\;\ge\; \Big(\sum_c \Pr_\theta(c)\Big)^2 = 1.
		\]
		Since $\sum_c y_c^2=\sum_c(\lambda_\theta\Pr_\theta(c)+(1-\lambda_\theta)\Pr_{\theta'}(c))=1$.
		Equality holds iff $x$ and $y$ are collinear, which here requires
		$\Pr_\theta(\cdot)=\Pr_{\theta'}(\cdot)$. Thus, when $\Pr_\theta(\cdot)\neq \Pr_{\theta'}(\cdot)$,
		the inequality is strict and the displayed average exceeds $\lambda_\theta$.
		
		\emph{(ii) Monotone heterogeneity.}
		Write the likelihood ratio $r(c):=\Pr_g(c)/\Pr_b(c)$.
		Then
		\[
		h_{g,c}(g,c)
		=\frac{\lambda_g r(c)}{\lambda_g r(c)+\lambda_b},
		\qquad
		h_{b,c}(b,c)
		=\frac{\lambda_b}{\lambda_b+\lambda_g r(c)}.
		\]
		Both expressions are monotone in $r(c)$, with $h_{g,c}(g,c)$ increasing and
		$h_{b,c}(b,c)$ decreasing. Under likelihood-ratio dominance, $r(c)$ is increasing in $c$,
		and the monotonicity claims follow.
		
		Moreover,
		$h_{g,c}(g,c)\ge \lambda_g$ iff $r(c)\ge 1$ (equivalently $\Pr_g(c)\ge \Pr_b(c)$), and
		$h_{b,c}(b,c)\le \lambda_b$ iff $r(c)\ge 1$. With $\bar c$ defined as above and $r(\cdot)$
		increasing, these inequalities hold for all $c\ge \bar c$ and reverse for $c<\bar c$,
		with strictness whenever $\Pr_g(c)\neq \Pr_b(c)$.\hfill\eproof
		
		\section{Detailed Dynamics}\label{app:detail}
		
		In this section, we describe the dynamics underlying partial homophily. Recall that we focus on the case with $1>c_g > p \geq c_b > 0.$ Under these conditions, agents with cost \(c_g\)  use the safe action as their default, switching to the risky action only when they receive information that updates their posterior belief to a sufficiently high probability that \(v=1\). Conversely, agents with cost \(c_b\)  have the risky action as their default.
		
		\subsubsection*{Signal Structure and Informative Actions}
		
		Agents with positive cost (i.e., with cost \(c_\theta > 0\)) send perfectly informative  payoff signs about the state \(v\). Specifically, a positive cost agent who receives a \emph{\(+\)} signal from any nonzero cost agent will choose the risky action. Similarly, a \emph{\(-\)} signal leads the agent to choose the safe action.

		If a signal profile lacks any \emph{\(+\)} or \emph{\(-\)} signals, it comprises uninformative \emph{\(+\)} signals from zero cost agents and \(\emptyset\) signals from agents who chose the safe action. Although these \(\emptyset\) signals do not lead to a direct inference about $v$, they still convey information about \(v\) through the observed behavior of nonzero cost agents.
		
		Because of this richer signal structure, the dynamics are more involved than in the one-dimensional case. For example, an initially optimistic blue agent might observe many green (and even some blue) agents taking the safe action. Such an observation can persuade the blue agent that the state is likely \(v=0\), causing them to abandon the default risky action.
		
		\subsubsection*{Signal Profiles and Posterior Beliefs}
		
		More precisely, consider an agent of group \(\theta\) in period \(t\) who observes a signal profile
		\[
		\mathbf{s}_\theta(n_b,n_g),
		\]
		where \(n_g\) cost \(c_g\) agents take the safe action, \(n_b\) cost \(c_b\) agents take the safe action, and the remaining \(d_\theta - n_b - n_g\) agents are of cost \(0\).

		Given the state \(v\), the probability that an agent of group \(\theta\) observes this signal profile in period \(t+1\) is determined by a multinomial formula. For example, the probability for a green agent is
		
		{\scriptsize
			\[
			\phi_g\bigl(\mathbf{s}_g(n_b,n_g) \mid g_t(v), b_t(v)\bigr) = \frac{d_g!}{n_b!\,n_g!\,(d_g - n_b - n_g)!} \Bigl[h_g\,\pi_g\,(1 - g_t(v))\Bigr]^{n_g} \Bigl[(1 - h_g)\,\pi_b\,(1 - b_t(v))\Bigr]^{n_b} \Bigl[h_g\,(1-\pi_g) + (1-h_g)\,(1-\pi_b)\Bigr]^{d_g - n_b - n_g}.
			\]
		}
		An analogous expression holds for blue agents.

		The posterior belief that \(v=1\) for an agent in group \(\theta\) upon observing the signal profile \(\mathbf{s}_\theta(n_b,n_g)\) is given by Bayes' rule:
		
		{\footnotesize
			\begin{equation} \label{posterior-expectation}
				\beta_\theta(n_b,n_g \mid g_t(\cdot), b_t(\cdot)) = \frac{p\, \phi_\theta\bigl(\mathbf{s}_\theta(n_b,n_g) \mid g_t(1), b_t(1)\bigr)}{p\, \phi_\theta\bigl(\mathbf{s}_\theta(n_b,n_g) \mid g_t(1), b_t(1)\bigr) + (1-p)\, \phi_\theta\bigl(\mathbf{s}_\theta(n_b,n_g) \mid g_t(0), b_t(0)\bigr)}.
			\end{equation}
		}
		
		When a positive cost agent takes the risky action, the resulting signal reveals \(v\) perfectly. Therefore, if the receiving agent observes at least one positive (or negative) signal, they immediately choose the risky (or safe) action. Absent such signals, an agent opts for the risky action if and only if the posterior belief in \eqref{posterior-expectation} exceeds the corresponding cost threshold.
		
		\subsection*{Simplification via Monotonicity}
		
		Recall that Lemma \ref{monotone} established that
		\(
		g_{t+1}(0) \leq g_{t+1}(1) \quad \text{and} \quad b_{t+1}(0) \leq b_{t+1}(1). \)
		Consequently, we showed that if an agent observes a safe action from a group \(\theta\) agent, the posterior belief shifts in favor of \(v=0\). That is, for any \(n_b\) and \(n_g\),
		\[
		\beta_\theta(n_b,n_g \mid g_t(\cdot), b_t(\cdot)) \leq \beta_\theta(0,0 \mid g_t(\cdot), b_t(\cdot)) = p.
		\]
		Because \(c_g > p\), the condition
		\( \beta_g(n_b,n_g \mid g_t(\cdot), b_t(\cdot)) > c_g \)
		is never met. Hence, \(\emptyset\) signals do not prompt green agents to deviate from their default safe action. The resulting dynamics are described by the following system:
		{\footnotesize
			\begin{align}
				g_{t+1}(0) &= 0,   \notag \\
				g_{t+1}(1) &= 1 - \Bigl[1 - \bigl(h_g\,\pi_g\,g_t(1) + (1-h_g)\,\pi_b\,b_t(1)\bigr)\Bigr]^{d_g},  \notag \\
				b_{t+1}(0) &= \sum_{\substack{n_b,n_g \geq 0 \\ n_b+n_g \leq d_b}}
				\phi_b\bigl(\mathbf{s}_b(n_b,n_g) \mid g_t(0), b_t(0)\bigr)\, \mathbbm{1}\Bigl\{\beta_b(n_b,n_g \mid g_t(\cdot), b_t(\cdot)) > c_b\Bigr\}, \notag\\
				b_{t+1}(1) &= 1 - \sum_{\substack{n_b,n_g \geq 0 \\ n_b+n_g \leq d_b}}
				\phi_b\bigl(\mathbf{s}_b(n_b,n_g) \mid g_t(1), b_t(1)\bigr)\, \mathbbm{1}\Bigl\{\beta_b(n_b,n_g \mid g_t(\cdot), b_t(\cdot)) < c_b\Bigr\}. \notag
			\end{align}
		}

		Note that,\(b_t(0), g_t(0)\) and  \(b_t(1), g_t(1)\) are interrelated through the indicator functions \(\mathbbm{1}\{\beta_\theta(n_b,n_g \mid g_t(\cdot), b_t(\cdot)) < c_\theta\}\). When the dynamics are  continuous (or differentiable), these indicators are locally constant. Thereby, the dynamics are locally decoupled whenever they are continuous (or differentiable).
		
		\subsection*{Comparative Statics}
		
		To analyze the comparative statics, we compute the partial derivatives around a neighborhood of \(\alpha\) (whenever \(g_{t+1}\) and \(b_{t+1}\) are differentiable with respect to \(g_t\) and \(b_t\)). For green agents, we have
		\begin{equation}
			\frac{\partial g_{t+1}(1)}{\partial g_t(1)} = d_g\,(1-A)^{d_g-1}\,h_g\,\pi_g > 0 \quad \text{and} \quad \frac{\partial g_{t+1}(1)}{\partial b_t(1)} = d_g\,(1-A)^{d_g-1}\,(1-h_g)\,\pi_b > 0, \label{green-partials}
		\end{equation}
		where
		\[
		A = h_g\,\pi_g\,g_t(1) + (1-h_g)\,\pi_b\,b_t(1) < 1.
		\]
		Similarly, for blue agents (when differentiable) we obtain
		
		\begin{align}
			\frac{\partial b_{t+1}(1)}{\partial g_t(1)} &=
			- \sum_{\substack{n_b,n_g \geq 0 \\ n_b+n_g \leq d_b}}
			\mathbbm{1}\Bigl\{\beta_b(n_b,n_g \mid g_t(\cdot), b_t(\cdot)) < c_b\Bigr\}
			\frac{\partial \phi_b\bigl(\mathbf{s}_b(n_b,n_g)\mid g_t(1),b_t(1)\bigr)}{\partial g_t(1)}
			\geq 0, \label{blue-partial-1}\\
			\frac{\partial b_{t+1}(1)}{\partial b_t(1)}
			&=
			- \sum_{\substack{n_b,n_g \geq 0 \\ n_b+n_g \leq d_b}}
			\mathbbm{1}\Bigl\{\beta_b(n_b,n_g \mid g_t(\cdot), b_t(\cdot)) < c_b\Bigr\}
			\frac{\partial \phi_b\bigl(\mathbf{s}_b(n_b,n_g)\mid g_t(1),b_t(1)\bigr)}{\partial b_t(1)}
			\geq 0.  \label{blue-partial-2}
		\end{align}
		The weak inequalities are the relevant conclusion; they are strict only when some locally safe-inducing no-direct profile contains the corresponding friend type with positive probability. The signs follow from the derivatives of the profile probabilities:
		\begin{align*}
			\frac{\partial \phi_b}{\partial g_t(1)} &= -\,n_g\,(1-h_b)\,\pi_g\,C_b\,A_b^{n_b}\,B_b^{n_g-1}\,R_b^{\,d_b-n_b-n_g} \leq 0 \quad \text{(negative if \(n_g>0\))},  \\
			\frac{\partial \phi_b}{\partial b_t(1)} &= -\,n_b\,h_b\,\pi_b\,C_b\,A_b^{n_b-1}\,B_b^{n_g}\,R_b^{\,d_b-n_b-n_g} \leq 0 \quad \text{(negative if \(n_b>0\))}.
		\end{align*} 
		where 
		$$C_b=\frac{d_b!}{n_b!n_g!(d_b-n_g-n_b)!}, \qquad A_b=h_b \pi_b(1-b_t(1)),$$ and $$B_b=(1-h_b) \pi_g\left(1-g_t(1)\right), \qquad R_b=h_b(1-\pi_b)+(1-h_b)(1-\pi_g).$$
		Finally, note that:
		\begin{align}
			\frac{\partial g_{t+1}(1)}{\partial h_g} &= d_g\,(1-A_g)^{d_g-1}\Bigl(\pi_g\,g_t(1)-\pi_b\,b_t(1)\Bigr), &  \frac{\partial b_{t+1}(1)}{\partial h_g} = 0 \label{partial-hg}
		\end{align} 
		where again  $A_g = h_g\,\pi_g\,g_t(1) + (1-h_g)\,\pi_b\,b_t(1)<1$.

		\subsection*{Existence of the Fixed Point}

		We are interested in the dynamic system described with the equations above, and denote the steady state of this system by $\alpha^*(\cdot)=(g^*(0),g^*(1),b^*(0),b^*(1))$. 
		
		Consider the decision of an agent receiving signal $\mathbf{s}_\theta$ given a previous-period state $\alpha\in[0,1]^4$,  and allow for mixing between taking the risky and safe actions. The expected-payoff objective is continuous in $\alpha$ for each signal profile.  Berge's Maximum Theorem therefore implies that the individual argmax correspondence is non-empty, compact-valued, and upper hemi-continuous. Convexity of the individual best-response set comes from allowing mixed actions, so the best-response set is an interval of mixing probabilities whenever the agent is indifferent and a singleton otherwise.
		
		The aggregate best response for agents with type $(\theta,c)$  is obtained by taking expectations over signal realizations. It inherits upper hemi-continuity and non-emptiness, and it is convex valued because it aggregates interval-valued mixed best responses over a continuum of agents. Kakutani's fixed-point theorem therefore gives a fixed point $\alpha^*\in\Gamma(\alpha^*)$, which is precisely a steady state of the dynamic system.

		Finally, any aggregate mixed equilibrium can be implemented in pure strategies at the population level: within $(\theta,c)$ types, all agents face the same signal distribution, and a mixing probability $q$ is implemented by having a fraction $q$ of that continuum group choose the risky action. The steady-state coordinates $g^*(v)$ and $b^*(v)$ are these aggregate group fractions.
		
		\section{Same-Default Cases under Partial Homophily}
		
		\label{app:same-default}
		
		In this appendix we provide the transition equations and steady-state conditions for the two cases in which both groups have the same default action. These cases are close to the
		full-homophily benchmark because both groups initially move in the same
		direction. Either both groups experiment under the prior, or both groups avoid experimentation under the prior. Cross-group links then change the probability of observing informative experimenters, but they do not create the asymmetric default behavior studied in the main text.
		
		Throughout this appendix, the second same-default case is stated for
		\(p<c_b,c_g<1\). This restriction ensures that if a positive cost agent takes
		the risky action, then success reveals \(v=1\) and failure reveals \(v=0\).
		
		For a green observer, let \(\mathbf{s}_g(n_b,n_g)\) denote a no-direct-signal
		profile in which \(n_b\) positive cost blue friends take the safe action,
		\(n_g\) positive cost green friends take the safe action, and all remaining
		observed friends have cost zero. Conditional on state \(v\), the probability of
		this profile is
		{\footnotesize
			\begin{align*}
				\phi_g\bigl(\mathbf{s}_g(n_b,n_g)\mid g_t(v),b_t(v)\bigr)
				&=
				\frac{d_g!}{n_b!\,n_g!\,(d_g-n_b-n_g)!} \\
				&\quad \times
				\Bigl[(1-h_g)\pi_b(1-b_t(v))\Bigr]^{n_b}
				\Bigl[h_g\pi_g(1-g_t(v))\Bigr]^{n_g} \\
				&\quad \times
				\Bigl[(1-h_g)(1-\pi_b)+h_g(1-\pi_g)\Bigr]^{d_g-n_b-n_g}.
			\end{align*}
		}
		For a blue observer, the analogous probability is
		{\footnotesize
			\begin{align*}
				\phi_b\bigl(\mathbf{s}_b(n_b,n_g)\mid g_t(v),b_t(v)\bigr)
				&=
				\frac{d_b!}{n_b!\,n_g!\,(d_b-n_b-n_g)!} \\
				&\quad \times
				\Bigl[h_b\pi_b(1-b_t(v))\Bigr]^{n_b}
				\Bigl[(1-h_b)\pi_g(1-g_t(v))\Bigr]^{n_g} \\
				&\quad \times
				\Bigl[h_b(1-\pi_b)+(1-h_b)(1-\pi_g)\Bigr]^{d_b-n_b-n_g}.
			\end{align*}
		}
		The posterior belief after a no-direct-signal profile is
		{\footnotesize
			\begin{align*}
				\beta_\theta(n_b,n_g\mid g_t(\cdot),b_t(\cdot))
				&=
				\frac{
					p\,\phi_\theta\bigl(\mathbf{s}_\theta(n_b,n_g)\mid g_t(1),b_t(1)\bigr)
				}{
					p\,\phi_\theta\bigl(\mathbf{s}_\theta(n_b,n_g)\mid g_t(1),b_t(1)\bigr)
					+
					(1-p)\,\phi_\theta\bigl(\mathbf{s}_\theta(n_b,n_g)\mid g_t(0),b_t(0)\bigr)
				}.
			\end{align*}
		}
		The value of \(\beta_\theta\) at zero-probability profiles is immaterial.
		
		\subsection*{Both groups have the risky action as their default}
		
		Suppose \(c_b,c_g\leq p\). Then positive cost agents in both groups take the
		risky action if they receive no information beyond the prior. If \(v=1\), a
		positive cost friend who takes the risky action generates a success signal,
		which reveals the good state. If \(v=0\), a positive cost friend who takes the
		risky action generates a failure signal, which reveals the bad state. Hence,
		direct signals are decisive. The only nontrivial part of the dynamics concerns profiles with no direct signal. The dynamics are
		{\footnotesize
			\begin{align*}
				g_{t+1}(1)
				&=
				1-
				\Bigl(
				1-h_g\pi_g g_t(1)-(1-h_g)\pi_b b_t(1)
				\Bigr)^{d_g} \\
				&\quad+
				\sum_{\substack{n_b,n_g\geq 0\\ n_b+n_g\leq d_g}}
				\phi_g\bigl(\mathbf{s}_g(n_b,n_g)\mid g_t(1),b_t(1)\bigr)
				\mathbbm{1}
				\left\{
				\beta_g(n_b,n_g\mid g_t(\cdot),b_t(\cdot))\geq c_g
				\right\}, \\[0.4em]
				g_{t+1}(0)
				&=
				\sum_{\substack{n_b,n_g\geq 0\\ n_b+n_g\leq d_g}}
				\phi_g\bigl(\mathbf{s}_g(n_b,n_g)\mid g_t(0),b_t(0)\bigr)
				\mathbbm{1}
				\left\{
				\beta_g(n_b,n_g\mid g_t(\cdot),b_t(\cdot))\geq c_g
				\right\}, \\[0.4em]
				b_{t+1}(1)
				&=
				1-
				\Bigl(
				1-h_b\pi_b b_t(1)-(1-h_b)\pi_g g_t(1)
				\Bigr)^{d_b} \\
				&\quad+
				\sum_{\substack{n_b,n_g\geq 0\\ n_b+n_g\leq d_b}}
				\phi_b\bigl(\mathbf{s}_b(n_b,n_g)\mid g_t(1),b_t(1)\bigr)
				\mathbbm{1}
				\left\{
				\beta_b(n_b,n_g\mid g_t(\cdot),b_t(\cdot))\geq c_b
				\right\}, \\[0.4em]
				b_{t+1}(0)
				&=
				\sum_{\substack{n_b,n_g\geq 0\\ n_b+n_g\leq d_b}}
				\phi_b\bigl(\mathbf{s}_b(n_b,n_g)\mid g_t(0),b_t(0)\bigr)
				\mathbbm{1}
				\left\{
				\beta_b(n_b,n_g\mid g_t(\cdot),b_t(\cdot))\geq c_b
				\right\}.
			\end{align*}
		}
		
		The corresponding steady state is
		\begin{align*}
			g(1)&=1, &
			b(1)&=1, \\[0.4em]
			g(0)
			&=
			\Bigl[
			h_g(1-\pi_g)+(1-h_g)(1-\pi_b)
			\Bigr]^{d_g}, &
			b(0)
			&=
			\Bigl[
			h_b(1-\pi_b)+(1-h_b)(1-\pi_g)
			\Bigr]^{d_b}.
		\end{align*}

		The logic is as follows. In the good state, all positive cost agents eventually
		experiment: observing any positive cost friend who experimented reveals
		\(v=1\), while observing only zero cost friends leaves the posterior equal to
		the prior \(p\), which is enough to justify the risky action because
		\(c_\theta\leq p\). In the bad state, any positive cost friend is informative. If the friend experimented, the observer sees a failure. If the friend did not experiment, that safe action is impossible in the good state at the steady state just described. Thus a positive cost observer experiments in the bad state only when all observed friends have cost zero. This gives the two probabilities displayed above.
		
		The full-homophily benchmark is the special case in which a green agent only
		observes green friends and a blue agent only observes blue friends. Then the
		expressions reduce to $g(0)=(1-\pi_g)^{d_g}$ and $b(0)=(1-\pi_b)^{d_b},$ with steady state \(g(1)=b(1)=1\).
		
		\subsection*{Both groups have the safe action as their default}
		
		Now suppose \(p<c_b,c_g<1\). Then positive cost agents in both groups take the
		safe action under the prior. By Lemma~\ref{monotone}, no-direct-signal profiles
		weakly lower the posterior relative to the prior, so they cannot induce
		experimentation when \(p<c_\theta\). Thus positive cost agents experiment only
		after observing a success by a positive cost agent. 
		
		In the bad state, no positive cost agent ever receives such a success signal, so $g(0)=0=b(0)$.

		In the good state, the steady-state conditions are
		{\footnotesize
			\begin{align*}
				g(1)
				&=
				1-
				\Bigl(
				1-h_g\pi_g g(1)-(1-h_g)\pi_b b(1)
				\Bigr)^{d_g}, \\[0.4em]
				b(1)
				&=
				1-
				\Bigl(
				1-h_b\pi_b b(1)-(1-h_b)\pi_g g(1)
				\Bigr)^{d_b}.
			\end{align*}
		}
		These equations always have the zero solution $g(1)=0=b(1)$. This is the no-experimentation steady state: if no positive cost agent
		experiments in the good state, then no later positive cost agent observes a
		success, and the safe action remains self-confirming.
		
		For the clean interior case \(\pi_b,\pi_g\in(0,1)\) and
		\(h_b,h_g\in(0,1)\), the existence of a nonzero steady state is governed by
		the linearization at the zero steady state. That linearization is the matrix
		{
			\begin{align*}
				\begin{pmatrix}
					d_g h_g\pi_g & d_g(1-h_g)\pi_b \\
					d_b(1-h_b)\pi_g & d_b h_b\pi_b
				\end{pmatrix}.
			\end{align*}
		}
		If the spectral radius of this matrix is weakly below one, then the zero steady state is the
		unique steady state. If it is strictly above one, then there is a unique
		strictly positive steady state $(g(1),b(1))\in(0,1)^2$ solving
		{
			\begin{align*}
				g(1)
				&=
				1-
				\Bigl(
				1-h_g\pi_g g(1)-(1-h_g)\pi_b b(1)
				\Bigr)^{d_g}, \\[0.4em]
				b(1)
				&=
				1-
				\Bigl(
				1-h_b\pi_b b(1)-(1-h_b)\pi_g g(1)
				\Bigr)^{d_b}.
			\end{align*}
		}
		The zero steady state is then unstable, while the strictly positive steady state is the stable learning steady state. 
		
		In the interior case \(h_b,h_g,\pi_b,\pi_g\in(0,1)\), the interior steady state is the unique stable steady state if and only if either $d_g h_g\pi_g\geq 1,$ or  $d_b h_b\pi_b\geq 1,$ or
		\begin{align*}
			d_gd_b(1-h_g)(1-h_b)\pi_g\pi_b
			>
			\bigl(1-d_g h_g\pi_g\bigr)
			\bigl(1-d_b h_b\pi_b\bigr).
		\end{align*}
		The first two inequalities say that one group can sustain learning through same-group observations alone. The last inequality covers the case in which neither group is self-sustaining in isolation, but the two cross-group learning channels jointly sustain experimentation.
		
		For arbitrary integer degrees,
		there is generally no useful closed-form expression for this positive steady state beyond these two polynomial equations. 
		
		This is the two-dimensional analogue of the full-homophily threshold. Under
		full homophily, the system separates into $g(1)=1-\bigl(1-\pi_g g(1)\bigr)^{d_g},$ and $b(1)=1-\bigl(1-\pi_b b(1)\bigr)^{d_b}.$ So, group \(\theta\) has a positive stable learning steady state exactly when
		\(\pi_\theta d_\theta>1\).
		
	\end{appendices}
\end{document}